\documentclass{article}
\usepackage{mathptmx}
\usepackage[mathscr]{euscript}
\usepackage{amsmath,amsthm,amssymb}
\usepackage{natbib}
\usepackage{ifsym}
\usepackage{graphicx} 
\usepackage{geometry}[width=0.2in]
\usepackage{authblk}
\usepackage{color}
\usepackage{xcolor}
\usepackage{comment}
\usepackage{hyperref}
\usepackage{url}
\usepackage{dirtytalk}
\usepackage{setspace}
\newtheorem{definition}{Definition}
\newtheorem{theorem}{Theorem}

\begin{document}

\title{Cooperative Game Theory Model for Sustainable UN Financing: Addressing Global Public Goods Provision}
\author[1,*]{Labib Shami}
\author[2,3]{Teddy Lazebnik}

\affil[1]{Department of Economics, Western Galilee College, Acre, Israel}
\affil[2]{Department of Information Systems, University of Haifa, Haifa, Israel}
\affil[3]{Department of Computing, Jonkoping University, Jonkoping, Sweden}
\affil[*]{Corresponding author: labibs@wgalil.ac.il}
\date{}
\maketitle

\begin{abstract}
This study introduces a novel cooperative game theory model designed to improve the United Nations’ current funding mechanisms, which predominantly rely on voluntary contributions. By shifting from a Nash equilibrium framework, where member states act in self-interest, to a cooperative model, the proposed approach aligns each country's financial contributions with the benefits they derive from UN activities. The model ensures a more sustainable and equitable system by introducing personalized pricing based on derived utility. Using agent-based simulations, the research demonstrates that the suggested approach increases global utility, reduces free-rider issues, and creates a more efficient resource allocation system. The findings suggest that the proposed model can optimize UN funding, ensuring a more stable and effective framework for global public goods provision, while considering the varying economic capacities of member states. Further research is recommended to assess the political viability of the model.
\end{abstract}

\section*{Introduction}
\label{sec:intro}
The financing mechanisms of international organizations are crucial for their effectiveness and operational viability. The United Nations (UN), as a leading global governance entity and multilateral organization, plays a vital role in addressing international challenges, from peace and security to humanitarian crises and development needs. Optimizing the UN's funding structure and ensuring its sustainability are crucial for enhancing its capacity to effectively address global challenges.

The UN relies on both the political will and funding of member states to implement its mandates and fulfill multilateral functions. Its financing mechanisms have evolved in an attempt to balance financial responsibilities among member states and adapt to changing global dynamics. The United Nations' funding landscape is characterized by a diverse array of sources, reflecting the organization's global reach and multifaceted responsibilities. These funding mechanisms, which have evolved over time, include mandated assessments, voluntary contributions, and trust funds. This multifaceted funding structure has profoundly influenced the UN's operational capabilities and ability to respond to complex international issues while navigating the intricacies of resource allocation and strategic prioritization.

Financial contributions are indispensable to the UN's core operations in peacekeeping, humanitarian aid, and development. Thus, inadequate funding can compromise aid deployment efficiency, potentially exacerbating human suffering in crises. Sustained funding is crucial for achieving long-term development goals such as poverty reduction, healthcare improvements, and educational advancements, while unpredictable funding constrains the UN's ability to plan and execute strategies effectively. As global issues become more complex, financial contributions from member states have grown increasingly precarious, raising concerns about funding sustainability and predictability. This necessitates a reassessment of funding strategies to enhance organizational resilience and effectiveness.

The inconsistent and sometimes inadequate financial contributions from member states create significant challenges across various operational areas. One of the critical issues surrounding the UN's funding is the increasing trend toward earmarked contributions, which can limit the flexibility of resource allocation. As noted by \cite{campos2018political}, the decentralized funding model allows donor countries to exert influence over specific programs, potentially skewing priorities away from broader UN objectives \citep{campos2018political}. This situation raises concerns about the sustainability of UN operations, as the reliance on voluntary contributions can lead to funding shortfalls for essential programs, particularly in humanitarian and development sectors \citep{browne2017vertical}. The earmarking of funds, while beneficial for specific projects, can undermine the principles of multilateralism that the UN is built upon, as highlighted by \cite{browne2017vertical}.

Moreover, the UN's funding challenges are exacerbated by geopolitical dynamics and the shifting priorities of member states. The reduction of contributions from major donors, particularly during the Trump administration, has sparked debates about the UN's financial viability and its capacity to maintain peacekeeping missions and humanitarian efforts \citep{abdenur2019peacekeeping}. The implications of such funding cuts are profound, as they not only affect immediate operational capabilities but also the long-term strategic goals of the organization. \cite{coleman2017extending} points out that reforming financing for UN peacekeeping is particularly challenging due to the politicized nature of these contributions. These financial uncertainties extend beyond peacekeeping, affecting crucial sectors such as humanitarian aid and development. The UN frequently depends on voluntary contributions, which can vary with changing political climates and donor priorities. This results in a fragmented financial environment, often leading to sporadic and insufficient funding windows that threaten both current and future programs.

Despite proposals for reforms aimed at stabilizing and increasing the predictability of funding, entrenched political and structural obstacles continue to impede progress. Tackling these challenges is essential for the UN to effectively fulfill its global mandates and achieve its long-term goals in peacekeeping, humanitarian aid, and development. To enhance its capacity to provide global public goods, the UN has considered negotiated pledges or replenishments as a funding mechanism that merges core funding features—particularly the delegation of autonomy—with the interests of member states in specific areas \citep{baumann2024financing, jenks2013united}. However, to date, negotiated pledges have only been implemented in the International Fund for Agricultural Development (IFAD) and the UN climate funds. 

This research introduces a novel funding approach to ensure sustainable UN operations financing, mitigating budgetary uncertainties. The proposed mechanism aims to optimize benefits for individual member states (individual utility) while maximizing the global impact of UN initiatives (global utility). It creates a stable financial ecosystem where member states are incentivized to maintain their commitments, as deviations would negatively impact individual and collective benefits. This is achieved by implementing a cooperative game theory model based on tailored pricing structures for each member state. This approach seeks to establish a robust, self-sustaining financial framework that aligns member states' interests with the UN's global objectives, fostering long-term stability and effectiveness in UN operations.

The provision of public goods such as healthcare, public education, or global public goods, is critical for governments and citizens alike \citep{wyrobek2022public, deneulin2007public, intro_1}. Theoretically, agents in an economy should prefer to allocate some, if not most, of their resources to public goods organized by an all-known agent \citep{intro_2}. In many cases, these theoretical predictions are not realized as they assume all agents in the economy are symmetric \citep{intro_5}. Indeed, in reality, agents express different types of asymmetry such as unique utility functions \citep{intro_3} or even the available information about the economy's state \citep{intro_4}. Recent studies focusing on the computational aspect of the economy even show different computational capabilities between agents \citep{intro_6,intro_7,intro_8}.

In a narrower context, the COVID-19 pandemic highlighted the critical need for global cooperation in addressing public health crises \citep{pandemic_important,who_data}. Given the transnational nature of infectious diseases, isolated national responses are insufficient \citep{close_1,close_2}. Indeed, \cite{close_bad} proposes that cooperative financing for public health goods, particularly in pandemic preparedness and response, is essential to mitigate the impact of such global challenges. The cooperative financing model involves determining individual contributions to public goods based on each participant's ability to pay and the benefits they receive. This approach mirrors personalized Lindahl pricing, ensuring a fair and efficient distribution of costs. To this end, pandemics require a global cooperative approach due to varying national health capacities \citep{better_together}. A cooperative financing model can equalize resource distribution for pandemic preparedness activities like surveillance, research, and healthcare infrastructure development. While such cooperation is specific and theoretically should include all countries globally, more sustainable cooperation with politically intervened countries, such as United Nations (UN) member states, is of great importance.

This paper depicts two well-known models for public goods allocation - the Nash and Lindahl equilibria. Although the Nash equilibrium outcome is conceptually appealing, it is not Pareto efficient \footnote{An allocation of resources is Pareto efficient if it is not possible to make anyone better off without making someone else worse off. However, if resources are allocated in an economically efficient manner, it does not imply equality or fairness.}. The possibility of social failure in the form of free riders encouraged many scholars to look for methods to overcome the inefficiency of the Nash equilibrium allocation \citep{shitovitz2001stable,peleg1986proof,coase1960problem}. Such a line of thought was developed by applying the Lindahl equilibrium conditions to the Nash allocation in an attempt to exploit the efficiency property inherent in the Lindahl cooperative equilibrium \citep{dijkstra2020pareto,chen2018collective,perets2012trading}. However, each of these studies presented limitations (for example, limiting the analysis to a single public good, symmetry among the economic agents, and a certain type of utility function) that prevented reaching general results regarding the existence and uniqueness of the allocation - both for confirmation or refutation of these properties. Recall that although the Lindahl equilibrium allocation in a public goods economy is efficient, it does not necessarily Pareto dominate the Nash equilibrium bundle \citep{shitovitz1998cournot}. This conclusion, which has been affirmed in many studies \citep{perets2012trading, shitovitz1998cournot, coase1960problem}, prevents individuals from preferring the Lindahl bundle over their Nash allocation.

To develop the model proposed in the study, we adopt \citeauthor{perets2012trading}'s (2012) model - the "\textit{Trading Equilibrium}" - which is based on a personalized price mechanism that yields a core allocation that strictly Pareto dominates the unique Nash allocation. The importance of the "Trading Equilibrium" is that it provides an efficient cooperative solution that Pareto dominates the non-cooperative Nash equilibrium. We develop a prescriptive model that takes the Nash Equilibrium as the starting point, under the assumption that the current financing model of the UN's activities (and in particular the amounts each country chooses to contribute beyond the amount it is obliged to transfer in practice) is obtained in accordance with the Nash equilibrium. From there, we consider a Pareto-improving path that relies on the principles of the Lindahl Equilibrium. However, unlike \citeauthor{perets2012trading}'s (2012) model with one private and one public goods, we consider a model with a finite number of consumption (private) and public goods. 

A generalization of the model to account for multiple private and public goods is essential when modeling the environment in which the United Nations operates, as it provides a range of global public goods. This phase of the research is critical, as it involves applying a model that reflects the economic framework governing member states' contributions to the financing of United Nations activities in the provision of global public goods. The incorporation of \citeauthor{perets2012trading}'s (2012) model, originally developed and validated for a single private and public good, is a pivotal step. The initial stage of this study will, therefore, establish the existence and uniqueness of the Trading equilibrium within an extended version of the model, which includes a finite number of both private and public goods. This extension is necessary to align the model with the research objectives, enabling its application to more complex scenarios involving the provision of global public goods by the various UN entities.

Based on the newly proposed model, we prove the existence of Nash and Trading equilibria allocations and show that their unity is guaranteed. Moreover, we demonstrate achieving the predicted Trading equilibrium using computer simulation and historical data of the sum of the contributions in each fund channel from each member country to finance UN activities. In particular, we show that the proposed model provides a computational framework to investigate and simulate international cooperation with countries that have asymmetric objectives and resources, making the proposed model highly realistic and implementable.

The rest of the paper is organized as follows. Section \ref{sec:related_work} provides an overview of the United Nation's funding model and the public goods economy theory. Section \ref{sec:model} provides a formal definition of the model as well as the proofs for the existence and uniqueness of the Nash and Trading equilibria allocations. Section \ref{sec:simulation} outlines a computer simulation that implements the proposed model for the UN funding. Finally, section \ref{sec:discussion} evaluates the framework's potential impact and concludes with recommendations for policy and future research directions.

\section{Related work}
\label{sec:related_work}
In this section, we provide the theoretical and applied background used in the proposed theoretical analysis and the computer simulation. Initially, we briefly review the structure of the UN's past and present funding models. Next, we present the development of multi-agent economics equilibria evolution from the Nash equilibria to the Trading equilibria.

\subsection{The United Nations' funding models and their development}

The United Nations' financial framework has undergone considerable transformation since its inception in 1945. Initially, the funding model combined voluntary contributions, mandatory assessments, and the establishment of trust funds, each of which played a vital role in supporting the organization’s diverse functions. This diversified funding approach sought to equitably distribute financial obligations among member states, thereby enabling the UN to function efficiently while maintaining broad political legitimacy.

The evolution of the UN’s early financing mechanisms underscores a continuous pursuit of financial stability and effectiveness. Assessed contributions were designed to provide a consistent and fair flow of resources, while voluntary contributions and trust funds offered the flexibility and targeted support necessary to address specific global issues. Analyzing these early funding structures reveals the inherent complexities and trade-offs in international financing, and underscores the ongoing necessity for innovative strategies to ensure the UN’s sustained operation in an evolving global landscape.

In its formative years, the UN primarily depended on assessed contributions, a mandatory financing model based on the "Scale of Assessments," which was calculated through formulas agreed upon by all member states. This scale allocated financial responsibilities according to member states' capacity to pay, accounting for factors such as per capita income and debt burdens \citep{haug2022international}. The Scale of Assessments reflects the principle of "differentiated universality," which seeks to balance global inclusivity and fairness by considering the unique economic circumstances of individual countries \citep{haug2022international}. This system was designed to ensure a stable and predictable flow of financial resources, representing a collective commitment to international cooperation and shared responsibility.

The United Nations’ reliance on assessed contributions proved inadequate in addressing the growing financial demands associated with its expanding array of activities, including peacekeeping, humanitarian assistance, and development programs. This financial shortfall led to the introduction of voluntary contributions, notably through initiatives such as the Expanded Program of Technical Assistance (EPTA). Unlike the regular budget, which is financed through mandatory assessments, programs under EPTA were supported by voluntary contributions from member states \citep{10.1093/oso/9780198877936.003.0004}. This dual funding model—combining mandatory and voluntary sources—enabled the UN to adjust to shifting political and financial contexts, allowing for more flexible responses to global challenges.

Nevertheless, voluntary contributions, while providing flexibility, also introduced complexities and dependencies. These contributions are discretionary, with member states determining both the amount and the purpose of their funding. A significant subset of voluntary contributions is earmarked funding, wherein donors designate how their funds are to be utilized. Wealthier nations, in particular, have contributed disproportionately, gradually shifting the financial burden onto a smaller group of countries \citep{10.1093/oso/9780198877936.003.0004}. While this ensured the continuation of critical programs, it increased the UN's susceptibility to financial instability and political influence from its largest donors.

The creation of trust funds marked another significant evolution in the UN's financial mechanisms. Trust funds, often established for specific purposes such as emergency relief or development projects, allowed for targeted funding and greater donor involvement. They reflected an adaptation to the increasingly complex global landscape and donor preferences. The rise of trust funds presents an opportunity for the UN to diversify its funding sources and reduce dependency on traditional donor contributions \citep{reinsberg2015rise}. 

Although the UN’s founders envisioned an organization primarily funded through assessed contributions, today the majority of its financial resources come from voluntary funding, which is provided at the discretion of donors who determine the amount, timing, and purpose of their contributions. Between 2011 and 2022, while the absolute amount of assessed contributions to the UN remained stable, their proportion within the overall funding structure has significantly decreased. In contrast, voluntary contributions, particularly earmarked funding, have seen notable growth in both absolute and relative terms, with other sources playing a relatively minor role. By 2022, assessed contributions accounted for only 18\% of the UN's funding, while voluntary contributions made up 82\% \citep{baumann2024financing}. This dramatic shift in the UN's financing model has had significant consequences for multilateralism. While voluntary funding has allowed the UN to expand its financial resources and activities beyond what would have been possible through assessed contributions alone, it has simultaneously diminished the influence of intergovernmental bodies governed by multilateral processes and altered the accountability frameworks of many UN programs and agencies \citep{graham2015money}.

The discretionary nature of voluntary funding undermines inclusive multilateral decision-making, as individual donors can dictate priorities by deciding when and how much to contribute, thereby influencing the UN's capacity to fulfill its mandates. Financing is not merely a matter of volume, but also of the mechanisms through which funds are transferred, as these have important political and organizational ramifications. Critics have raised concerns about whether the current financing model is appropriate, given the increasing focus on the UN's normative roles and the provision of global public goods rather than project-based activities. Some argue that the most suitable method to fund the UN's normative and global work is through assessed contributions, or alternatively through voluntary core contributions. However, the growing trend of financing global functions through selective voluntary contributions from a small group of donors, who often pick and choose which norms to support, reflects an ongoing challenge in the UN's financial structure \citep{jenks2013united}.

In contrast to the current funding structure of the United Nations, which relies heavily on voluntary contributions, the proposed model is based on obligatory contributions from each member state, determined according to the specific benefits they derive from UN activities. This approach not only ensures greater financial stability but also promotes voluntary participation by aligning contributions with the perceived value of UN services. Each UN agency's activities are framed as global public goods, with the value each nation assigns to these activities reflected in the portion of its budget allocated to the services provided. The specific contribution, or "personal price," that each member state is obligated to provide is derived from a cooperative public goods model, which will be elaborated upon in the subsequent section. This model aims to establish a more equitable and efficient financial framework for the provision of global public goods.

\subsection{The provision of public goods}
In an economic framework characterized by private ownership of all commodities and wherein economic agents are uniformly characterized as "small," the first welfare theorem posits that the competitive equilibrium allocation achieves Pareto efficiency, subject to the condition of local non-satiation \citep{arrow1951extension}. Nevertheless, in economies featuring public goods, the efficiency of a competitive equilibrium is not assured. Conventional neoclassical economic theory anticipates a suboptimal provision of public goods, attributing this to individuals engaging in free-riding behavior by leveraging the contributions of others. However, a voluntary public good contribution is observed in both the real world and laboratories \citep{mcginty2013public}.

Based on Paul Samuelson’s work, a public good is non-excludable and non-rivalrous, where no one can be excluded from its use and where the use by one does not diminish the availability of the good to others \citep{samuelson1954pure}. This concept has been debated and modified over time. In The Affluent Society, John Kenneth Galbraith suggests that public goods are “things [that] do not lend themselves to production, purchase, and sale. They must be provided for everyone if they are to be provided for anyone, and they must be paid for collectively or they cannot be had at all.” \citep{galbraith1998affluent}. Given the non-excludable and non-rival nature of public goods, they cannot be provided satisfactorily through a market mechanism but have to be provided through some form of public action (e.g. via taxation). The public provision does not necessarily entail government provision and public goods can be provided by other actors than governments.

The most common equilibrium concepts used in the theory of public goods economies are Nash and Lindahl equilibria. Both solution concepts are applied to the well-known economic problem of the private provision of public goods \citep{danziger1976graphic, cornes2007aggregative}. In Lindahl equilibrium, determining the level of public goods provided and their financing is solved by adapting the price mechanism so that an efficient allocation remains the result of voluntary market activities. By using the Lindahl approach, individuals face personalized prices (rather than some political choice mechanism or coercive taxation) at which they may purchase public goods. When the system of personalized prices is such that everyone prefers the same 
quantity of the public good, then the economy is in a Lindahl equilibrium. Since each individual buys and consumes the total levels produced of public goods, the price to producers is the sum of the prices individuals pay. Hence, in equilibrium, the supply at these prices equals common demand. Thus, Lindahl equilibrium assumes that the economy is cooperative and brings unanimity about public goods provision, with costs being shared proportionately to (marginal) benefits \citep{shitovitz1998cournot, roberts1989lindahl}.

The Nash equilibrium suggests a solution that is based on the assumption that the economy is non-cooperative. The amount of each public good purchased by any agent is determined independently. By taking the other consumers' purchases as given, each agent adjusts its purchase of the public goods to maximize utility. When each consumer is in a position where he has no incentive to change his purchase the economy is in Nash equilibrium. Because of consumers' myopic interests, one of the key results in the public good literature indicates that if there are at least two individuals, the Nash allocation is not Pareto-optimal \citep{danziger1976graphic}. The inefficiency in non-cooperative games is mainly due to the existence of two consequences of market failures: the first is known as "The Free Rider Problem" which emphasizes the fact that the larger the number of agents, the more each agent is convinced that reducing his contribution will not reduce the size of the public good since the other agents will continue to contribute as before, regardless of the size of his contribution, and it is even possible that the agent will not contribute at all to the production of the public good and will only benefit from its consumption \citep{hampton1987free, kim1984free}. The second market failure is the "Tragedy of the Commons" which was first introduced by \cite{hardin1968the} and deals with the conflict that arises between individuals in society regarding the provision of the common good when limited resources are exploited by those individuals. This leads to the inefficient and unsustainable use of resources, as individuals do not consider the long-term consequences of their actions. Non-cooperative individual actions lead to the overuse of common pool resources and under-provision of public goods. This ultimately leads to a depletion of resources that is not in the long-term interest of society. 

Thus, as \cite{coase1960problem} pointed out, when the economic agents behave according to the Nash equilibrium, they will look for trading opportunities that will generate efficient allocations which strictly Pareto dominates their (inefficient) Nash allocation. As \cite{samuelson1954pure} showed, one could draw a demand or marginal rate of substitution, curve for each individual, and that their vertical summation is the relevant object to equate with the marginal rate of transformation between public and private goods at a Pareto optimum. That is, for each public good one has \( \sum_{i \in \mathscr{N}} MRS_{Y_l-x_i}^{i}=RPT_{Y_l-x_i}\), where \(RPT_{Y_l-x_i}\) is the rate of technological substitution in production between public (\(Y_l\)) and private (\(x_i\)) goods. Since the Lindahl allocation is a vector of individual prices that satisfies the Samuelson criterion, it is always Pareto-optimal \citep{foley1970lindahl}. Nevertheless, the Lindahl allocation, while Pareto efficient, does not Pareto dominate (utility-wise) the Nash allocation, since some agents might strongly prefer their Nash consumption bundle over their Lindahl consumption bundle (see \cite{shitovitz1998cournot} examples 1 and 2 on page 6). 

Scholars have long compared the performance of these two equilibria in a public goods economy \citep{dijkstra2020pareto, danziger1976graphic}. With an explanatory purpose, \cite{ danziger1976graphic} graphically compared the non-cooperative Nash equilibrium and the cooperative Lindahl equilibrium, with the understanding that the transition from the Nash to the Lindahl equilibrium can be considered as caused by an income and a substitution effect. According to his findings, in a sufficiently large economy, the transition from a Nash to a Lindahl equilibrium will benefit everybody. However, in small economies, an ethical evaluation of the scattering of advantages and disadvantages may be necessary before deciding whether or not the Lindahl equilibrium should be preferred to the Nash equilibrium.

In a seminal paper, \cite{champsaur1975share} provided proof that the core of an economy with one private good and one public good is a vN\&M stable set. Based on this result and the work of \cite{peleg1986proof},  \cite{shitovitz2001stable} showed that in the same model setup, there exists a core allocation that Pareto dominates the inefficient Nash allocation with a finite number \(n \ge 2\) of consumers. This result was generalized in \cite{perets2012trading} to such an economy with a mixed measure space of consumers. The authors introduced a new equilibrium concept called "Trading Equilibrium", which is based on a personalized price mechanism that yields a core allocation which strictly Pareto dominates the unique Nash allocation in that market. In an economy with one private good, one public good, and linearly additive production technology of the public good, and by using the Nash allocation as the consumers' initial bundle (an approach similar to \cite{dijkstra2020pareto} and \cite{ chen2018collective}), the authors showed that the consumers will have the incentive to trade among them to achieve a different allocation which is efficient and Pareto dominates the one they currently have. The importance of the "Trading Equilibrium" is that it provides an efficient cooperative solution that Pareto dominates the non-cooperative Nash equilibrium. Unlike \citeauthor{danziger1976graphic}'s (1976) descriptive research method, these researchers used mathematical tools to obtain their results. However, their economy contained only one private and one public good. In the present study, we will relax this limitation and assume a finite number of private and public goods.

An interesting approach to the comparison between the models (albeit in a cooperative version of Nash equilibrium) was proposed by Dijkstra and Nentjes (2020) in the area of environmental economics. The authors compare the Exchange-Matching-Lindahl (EML) solution (a bottom-up mechanism) and the Nash Bargaining solution (a top-down mechanism) for the provision of one public good. Under the EML solution, countries are offered an exchange rate that specifies the ratio between global and national abatement. In equilibrium, exchange rates are such that each country demands the same amount from the world. This amount is the sum of all countries' supply, and the equilibrium is Pareto-efficient. Their results indicate that in a setting with two agents both mechanisms are equivalent. However, in a setting with more than two agents, EML benefits all agents. Unlike the current study, the researchers' work was limited to one public good and a quadratic utility function only.

As suggested in the studies reviewed, the use of cooperative game theory can help to address the issue of inefficient allocation in economies characterized by voluntary self-provision of public goods. However, this approach is not without its own challenges. Cooperative game theory can be difficult to implement in practice \citep{danziger1976graphic}, and the models can be complex and difficult to solve. Despite these challenges, the use of cooperative game theory can help to improve the efficiency of the provision of global public goods. This is an important issue, as global public goods are essential for the functioning of society. 

\section{Public goods economy with asymmetric agents}
\label{sec:model}
Despite the Nash equilibrium allocation's lack of Pareto efficiency, economic agents, once having attained it, exhibit no incentive to deviate given the choices of their counterparts. This is where the Trading equilibrium intervenes. It takes the Nash allocation as a starting point as if it were the initial bundle for the agents and sets personalized prices for each public good. These prices are tailored to each agent in a manner conducive to facilitating trade, thereby engendering a utility increase for at least one agent while maintaining indifference for others. Consequently, the Trading equilibrium supersedes the Nash allocation in Pareto dominance. 

Intuitively, consider a population of agents that can purchase (or contribute to) two types of goods: private and public. The former contributes solely to the utility of the agent, while the latter contributes to the utility of all agents within the population. All the agents are rational and all-knowing. The overarching objective of the population, and by extension, the economy, is to converge towards a stable equilibrium state.  

Formally, consider a pure public good economy with a finite-size set of agents (states), \(\mathscr{N} := \{1,...,N\}\), which can purchase \(\mathscr{K} := \{1,...,K\}\) private goods and contribute to the production of \(\mathscr{L} := \{1,...,L\}\) pure global public goods. Each agent, \(i \in \mathscr{N}\), is endowed with a pre-defined strictly positive amount \(\omega_i=(\omega_i^1,...,\omega_i^K)\in \mathbb{R}_{++}^K\) of the private goods. Let the vectors \(x_i=(x_i^1,...,x_i^K)\in \mathbb{R}_+^K\) represent units of the private goods each consumer, \(i\in \mathscr{N}\), decided to consume from his initial endowment. We assume that \(0\le x_i \le \omega_i\) for all consumers. Each agent, \(i \in \mathscr{N}\), has a unique utility function \(u_i(x_i, Y): \mathbb{R}^{K+L} \rightarrow \mathbb{R}\), where \(y_i=(y_i^1,...,y_i^L)\in \mathbb{R}_+^L\) is the vector that represents units of global public goods that were produced by the contribution of agent \(i\in \mathscr{N}\) from the private goods in his possession. Here, any public good can be produced by any private good on a one-to-one basis with a linear production function. Thus, the vector \(Y=(Y_1,...,Y_L)\) represents the total amount of the global public goods produced, such that \(Y_l := \sum_{i\in \mathscr{N}} {y_i^l}\). Namely, each agent's utility is determined by the amount of private goods and the aggregated agents' contribution toward the production of public goods. In relation to the research topic, expenditures on private goods correspond to the budgets of individual government offices within each country, while contributions toward the production of public goods represent the obligatory financial contributions from member states to support the operations of various UN agencies.

In addition, let us assume that the utility functions of the agents, \(\forall i \in [1, \dots, N]: u_i(x_i, Y)\), are continuous, strictly quasi-concave, and that agents derive positive marginal utility from the private and public goods. This assumption is used to make sure the utility functions are aligned with these empirically measured. Hence, all commodities are strictly desired over \(\mathbb{R}_{++}^{K+L}\) and agents are indifferent between all bundles on the boundary of \(\mathbb{R}_+^{K+L}\). Furthermore, we assume Ordinal Separability of the utility functions and the Strict Ordinal Normality for private and public goods (That is, increasing the amount of the private good or decreasing the amount of the public good will lead to an increase in the value of the marginal rate of substitution)
\begin{equation}
\frac{\partial MRS_{Y_l-x_i^k}^i}{\partial Y_l}<0 \; \wedge \; \frac{\partial MRS_{Y_l-x_i^k}^i}{\partial x_i^k}>0, \; s.t. \; MRS_{Y_l-x_i^k}^i = \frac {\frac{\partial u_i(x_i,Y)}{\partial Y_l}} {\frac{\partial u_i(x_i,Y)}{\partial x_i^k}}.
    \label{eq:mrs}
\end{equation}

Based on this model definition, a rational agent aims to optimize its own utility function. One can identify two types of equilibria: the Nash and Trading equilibria. The utility does not originate in the initial allocation as it is drawn from the "strategy" which is the division of the initial allocation between the purchase of the private goods and the contribution to the production of the public goods. The Nash (or the voluntary) provision of pure public goods is characterized by the non-cooperative game theoretic notion of Nash equilibrium. Thus:
\begin{definition}\label{Nashdef}
The Nash equilibrium in our economy is a \(N\)-tupple strategy \(s^*=(s_1^*,...,s_N^*)\)  such that:
\begin{enumerate}
\item \label{Nash1}
For each \(i\in \mathscr{N}\), \(s_i^*=(x_i^*,y_i^*)\in \mathscr{S}_i\) where:  \begin{equation*}
\mathscr{S}_i= \left. \begin{cases}  (x_i,y_i)\in \mathbb{R}_+^{K+L}:& \begin{array}{r@{}} \sum_{k\in \mathscr{K}} x_i^k+\sum_{l\in \mathscr{L}}y_i^l \le \sum_{k\in \mathscr{K}} \omega_i^k \\ 
x_i^k \le \omega_i^k \qquad \forall k\in \mathscr{K}\end{array} \end{cases}\right \}
\end{equation*}
 \item \label{Nash2}
 For each \(l\in \mathscr{L}\) we have \(Y_l^*=\sum_{i\in \mathscr{N}}y_i^{l^*}\) 
 \item \label{Nash3}
 For each \(i\in \mathscr{N}\) we have \(U_i(s^*) \ge U_i(s_i,s_{-i}^*)\)
\end{enumerate}
Where \(U_i(s^*)=U_i(s_1^*,...,s_N^*)=u_i(x_i^*,\sum_{j \in \mathscr{N}}(\omega_j - x_j^*))\). 
\end{definition}

Importantly, since the private and public goods in the proposed framework are by themselves symmetric in nature, we treat the allocation as a set rather than a vector. As such, symmetric allocations are treated as identical. For example, let us assume a single private good and two public goods, indicated by \(x^1, Y_1, Y_2\). For agent \(i\), both strategies \((x_i^1, y_i^1, y_i^2) = (1, 3, 2) \text{ and } (1, 2, 3)\) are identical if and only if the utility of the agent from both strategies is identical. This result is possible when both strategies lead to the same final allocation with the same amount of each of the public goods. Hence, there can be a unique Nash equilibrium \textit{allocation} with an infinite number of \textit{strategies} leading to that allocation. For more details, see the example in the appendix.

In a similar manner, the \textit{Trading equilibrium} of the proposed model is defined as follows:
\begin{definition}\label{Tradedef}
The Trading Equilibrium in our economy is an allocation  \((\bar{\bar{x}}_i,\bar{\bar{Y}})\) and a non-negative personalized
price vector \(\bar{\bar{p}}_i=(\bar{\bar{p}}_i^1,...,\bar{\bar{p}}_i^L)\) such that:
\begin{enumerate}
\item \label{Trade1}
\(\sum_{k\in \mathscr{K}} \sum_{i\in \mathscr{N}} \bar{\bar{x}}_i^k+ \sum_{l\in \mathscr{L}}\bar{\bar{Y}}_l=\sum_{k\in \mathscr{K}} \sum_{i\in \mathscr{N}} x_i^{k*}+\sum_{l\in \mathscr{L}} Y_l^* =\sum_{k\in \mathscr{K}} \sum_{i\in \mathscr{N}} \omega_i^k \)

\item \label{Trade2}
 \(\sum_{k\in \mathscr{K}} \bar{\bar{x}}_i^k+ \sum_{l\in \mathscr{L}}\bar{\bar{p}}_i^l \bar{\bar{Y}}_l\le  \sum_{k\in \mathscr{K}} x_k^*+\sum_{l\in \mathscr{L}} \bar{\bar{p}}_i^l Y_l^*\) For each \(i\in \mathscr{N}\) 
 
 \item \label{Trade3}
 For each \(i\in \mathscr{N}\) and for each \((x_i,Y)\), that satisfies  \(\sum_{k\in \mathscr{K}} x_i^k+ \sum_{l\in \mathscr{L}}\bar{\bar{p}}_i^l Y_l\le  \sum_{k\in \mathscr{K}} x_k^*+\sum_{l\in \mathscr{L}} \bar{\bar{p}}_i^l Y_l^*\), we have \(u_i(\bar{\bar{x}}_i,\bar{\bar{Y}})\ge u_i(x_i,Y)\)
 
 \item \label{Trade4}
 For each \(l\in \mathscr{L}\) we have \(\sum_{i\in \mathscr{N}} \bar{\bar{p}}_i^l=1\)
\end{enumerate}
\end{definition}
Recall that \((x_1^*,...,x_K^*;Y_1^*,...,Y_L^*)\) represents the Nash allocation. Moreover, by the desirability of the commodities we have in
condition \ref{Trade2} of definition \ref{Tradedef} budget equality (Thus, the budget available for the Trading equilibrium bundle for an agent is the monetary value of his Nash equilibrium bundle). Following these two definitions, we show that the Nash and Trading equilibria exist and are unique. 

\begin{theorem}
For the proposed economy, a unique Nash allocation exists.
\end{theorem}

\begin{proof}
In order to show that Nash allocation exists, we need to show that there is an allocation of all the agents in the population such that it is not beneficial for any agent to alter its allocation. Formally, let us define a set of functions \(\mu_i(S) \rightarrow \mathbb{R}\) that get the allocation of the entire population and return the utility of the \(i_{th}\) agent. Following the original proof by Nash himself, in order to prove the existence of a Nash equilibria, we use the Kakutani fixed-point theorem \citep{yu2016equivalence, nash1950equilibrium}. To this end, we first have to set \(Z\) and \(\Phi\) where \(Z\) is a non-empty, compact, and convex subset of some Euclidean space \(\mathbb{R}^n\) and \(\Phi: Z \rightarrow 2^Z\) such that \(\Phi\) has a closed graph and is non-empty and convex for all \(x \in Z\). For our economy, \(S\) is defined by the intersection of the amount of private and public goods each agent in the economy can buy - \(\forall i \in N: x_i + y_i \leq \omega_i\). Under the assumption, at least one agent, \(i \in N\), has resources; \(Z\) is non-empty. In addition, \(Z\) is the intersection of linear constructions that results in a compact and convex subset of \(\mathbb{R}^n\). The conditions for \(\Phi\) are also met according to the economy's agents' utility functions assumption. As such, \(\mu := \sum_{i \in N} \mu_i(S)\), has a fixed point, and therefore the economy has a Nash equilibrium. 

In addition, we wish to show that the Nash equilibrium has a unique allocation.
Let us falsy assume there are two Nash equilibria, \(s^*_1\) and \(s^*_2\) such that \(s^*_1 \neq s^*_2\). Following the third condition in the Nash equilibrium definition, for every agent in the population, \(i \in N\), we must have \(U_i(s^*_1) \geq U_i(s_i, s^*_{-i, 1})\) and \(U_i(s^*_2) \geq U_i(s_i, s^*_{-i, 2})\). In particular, \(s^*_1\) should satisfy the second condition and the other way around. As such, setting both conditions, we obtain that \(U_i(s^*_1) \geq U_i(s^*_2)\) and \(U_i(s^*_2) \geq U_i(s^*_1)\) which leads to \(U_i(s^*_1) = U_i(s^*_2)\). However, this contradicts our initial assumption that \(s^*_1 \neq s^*_2\) and therefore there is only one Nash equilibrium. For a detailed proof, please refer to the appendix.

\end{proof}

\begin{theorem}
For the proposed economy, a Trading allocation exists and is unique inside the allocation space.
\end{theorem}

\begin{proof}
In order to show a Trading equilibrium exists, one is required to find an equilibria with a global utility higher than or identical to the one obtained at the Nash equilibrium. To this end, as each agent in the population proposes searches for the global optimal, we will use the Simplex theorem to show there is an optimum to the objective function \citep{simplex}. In order to use the Simplex theorem we need to show that the objective function is continuous and derivative to all the parameters and at all possible allocations, and the constraints define finite, close, and concave. To this end, the global utility function \(u := \sum_{i \in N} u_i\) is a sum of continuous and strictly quasi-concave functions which means it is also continuous and strictly quasi-concave, as required. Next, the constraints of the Simplex task we define are linear as the sum of contributions of each agent to all the private and public goods is limited by the initial amount of resources it has. We also add the fourth definition of the Trading equilibrium, namely that the sum of prices is 1 as a constraint, which is yet another linear constraint. In addition, by definition, agents can not negatively contribute to either a private or public good which enforces a non-negativity constraint for each agent and for each good. Now, following the second and third conditions for the Trading equilibrium, we need to make sure that the Trading equilibria allocation of each agent is at least as good as one of the Nash equilibria allocations of the same agent. To this end, for each agent \(i \in N\) we add the constraint \(u_i(\bar{\bar{x}}_i,\bar{\bar{Y}})\ge u_i(x_i,Y)\). Since \(u_i\) is assumed to be concave, the intersection of these constraints is still concave as well. Overall, the intersection of constraints defines a finite, close, and concave space. Hence, the proposed economy satisfies the Simplex theorem condition and has a global optimum. This allocation is, by definition, a Trading equilibrium of the proposed economy, showing its existence.  

Now, in order to show that the Trading equilibrium is unique, we can use the \textit{extreme value} theorem \citep{labib_ask} while taking into consideration that \(\forall i \in N: U_i\) is monotonically increasing for every private and public good. Namely, let us falsy assume that there are two trading equilibria, \(t^*_1\) and \(t^*_2\). As both are global optimum of \(\sum_{i \in N} U_i\), according to Simplex theorem this can exist if and only if \(t^*_1\) and \(t^*_2\) are on the edge of the possible allocation space. As such, if one ignores the edge of the allocation space, the function \(\sum_{i \in N} U_i\) is a monotonic and harmonic function. Hence, it obtains a single optimum. 
\end{proof}

We provide more detailed proof of the uniqueness of the Nash equilibrium allocation that closely follows the definition in the appendix. In addition, a numeric example of how one can calculate the Nash and Trading equilibria is also provided in the appendix.

\section{Computer simulation}
\label{sec:simulation}
The financial contributions of member states to the United Nations provide a valuable case study for applying our global public goods provision model with asymmetric agents (states) \citep{pandemic_cool}. The intricate dynamics between states, each characterized by unique economic conditions, political frameworks, and objectives, closely resemble the heterogeneous agent structure inherent in our model. Furthermore, these contributions highlight the critical differences between Nash equilibria and cooperative (trading) equilibria. A purely self-interested strategy, similar to a Nash equilibrium, may result in suboptimal global outcomes, such as a UN initiative benefiting one nation disproportionately while being inconsequential to others. Conversely, a cooperative framework, wherein states work collectively to optimize global utility, can promote more balanced resource allocation and shared economic benefits. Our model’s mechanism of personalized Lindahl prices finds resonance in international agreements designed to allocate contributions based on a nation's economic capacity, preferences, and anticipated benefits. 

To operationalize the proposed model, we developed a computer simulation using Python (version 3.9.2), leveraging historical data to simulate a trading equilibrium. Initially, we assume UN member states historically followed a Nash equilibrium, focusing on national interests with limited regard for broader global outcomes. Utilizing this as a baseline, the model seeks to transition towards the cooperative Trading equilibrium. For the case study, we compiled the contributions of 138 UN member states (for which we have in our possession cross data for both budgets for government ministries and donations to the UN) in U.S. dollars across 10 UN offices (WFP, UNHCR, UNICEF, DPO, IOM, WHO, UN, FAO, UNDP, and UN-DPO)\footnote{The data is available at: \url{https://unsceb.org/fs-revenue-government-donor}}. Additionally, each country was assumed to operate 10 local offices serving as private goods, benefiting solely the nation-state\footnote{The data is available at: \url{https://data.un.org/Data.aspx?d=SNA&f=group_code\%3a301}}. A sum of a linear and logarithmic function based regression model was fitted to estimate each country's resources based on the total used budgets of both UN and local offices (adjusted by the average conversion rate for each year), providing projections for 2024.

Specifically, each country’s unique utility function (\(u_i\)) was estimated by correlating the portion of its budget allocated to different offices with the respective importance it assigns to the services provided. In other words, the utility weight that each country assigns to the services of a government ministry is inferred from the budget allocated to that ministry, fitting the function \(\alpha_1 x + \alpha_2 ln(x)\) where \(\alpha_1, \alpha_2 \in \mathbb{R}^+\) are parameters and the function is aligned with the condition in the proposed model and the law of diminishing return \citep{shephard_fare_1974}. Similarly, the utility weight attributed to each United Nations branch, reflecting the value of the global public good, is derived from the monetary contributions made by each country over time, fitted to the same function. This method, using revealed preferences, helps determine the coefficients in each country’s unique utility function for both national private goods (government ministry services) and global public goods (UN services). This results in the weights \(\zeta_j^i, \xi_{c,1}^i, \xi_{c,2}^i \in \mathbb{R}^+\) for \(j \in [1, \dots, 10], c \in [1, \dots, 1380]\). Formally, the utility function of country \(i\) (for a representative person within the country) takes the form:
\begin{equation}
    u_i := \sum_{j=1}^{10} ( \zeta_{j}^i Y_j ) + \sum_{c=1}^{1380} (\xi_{c,1} x_c + \xi_{c,2} ln(x_c)).
\end{equation}
Notably, there are 10 global public goods in the form of UN offices and each country (138 of them) has 10 local offices that operate as private goods, resulting in 1390 goods in total for this configuration. For country \(k\), \(\xi_{c_1}, \xi_{c,2} \neq 0\) such that \(i \in [10 (k + 1), 10 (k + 1) + 9]\) and \(\xi_{c,1} = \xi_{c,1} = 0\), otherwise. Importantly, this utility function satisfies the proposed model's assumptions. 

We implemented this economic model in a computer simulation using the agent-based simulation approach \citep{agents_1,agents_2,agents_3} where each agent (state) in the population is represented using a finite state machine \citep{agent_finite_state_machine} that contains its available resources as well as its current allocation. The simulation operates in rounds, such that each round is denoted by \(t_i\) such that \(t_i \in [0, T]\) where \(T < \infty\). First, at \(t_0\), the initial amount of resources and utility functions of the agents in the population are allocated. Afterward, at each round \(t_i\), for each agent (state) in the population, it computes its optimal strategy concerning its utility function and assumes all other agents are available to it. The agents update their strategy at the end of each round after all the agents in the population compute their new strategy. This process repeats until at some round, \(t_N\), all agents in the population keep the same strategy as in the \(t_{N-1}\) round. At this point, we iteratively reach the Trading equilibrium using the gradient descent (GD) method \citep{gradient_descent_method}. The GD stops when the gradient's norm is smaller than a pre-defined threshold \(\epsilon \in \mathbb{R}^+\), which indicates the economy reached the Trading equilibrium. 

Fig. \ref{fig:converage_rate} shows the results of this simulation for \(n=100\) initial conditions with \(T = 200\) and a population of size \(N = 138\) (number of member states for which we have data) where the x-axis indicates the simulation round and the y-axis is the normalized average global utility (the absolute values hold limited significance; therefore, we normalize them relative to the final objective, namely the Trading Equilibrium, to derive a meaningful measure of relative improvement). The blue (solid) line indicates the average case in which all the other cases are normalized with respect to it. The green vertical line indicates the average duration taken to coverage, resulting in the Trading equilibria. 

As demonstrated in Figure 1, the implementation of the proposed funding model yields a significant increase in global utility. The initial point reflects the global utility derived from the current funding model utilized by the United Nations, which is based mainly on voluntary contributions rather than mandatory payments. In this scenario, the global utility is valued at 0.941 points, following normalization to the Trading equilibria, representing the Nash equilibrium. However, the global benefit under the proposed trading equilibrium model shows about a 6\% increase. This represents an improvement in global utility. Moreover, the transition to the proposed model results in a marked increase in the utility for individual states, with an increase in the value of the utility ranging from 2.5\% to 7.5\% \footnote{For a comparison between the two funding models for all member states, and for a more detailed comparison for selected countries, please refer to the supplementary materials.}. It is important to note that the transition to the new financing model did not result in increased contributions from all member states. Notably, two of the largest contributors to the United Nations—the United States and China—experienced a slight reduction in their required payments. Nevertheless, both their individual benefits and the overall global benefit were enhanced under the new model, demonstrating the efficiency and equity gains achieved through the revised funding mechanism. 

\begin{figure}[!ht]
    \centering
    \includegraphics[width=0.99\textwidth]{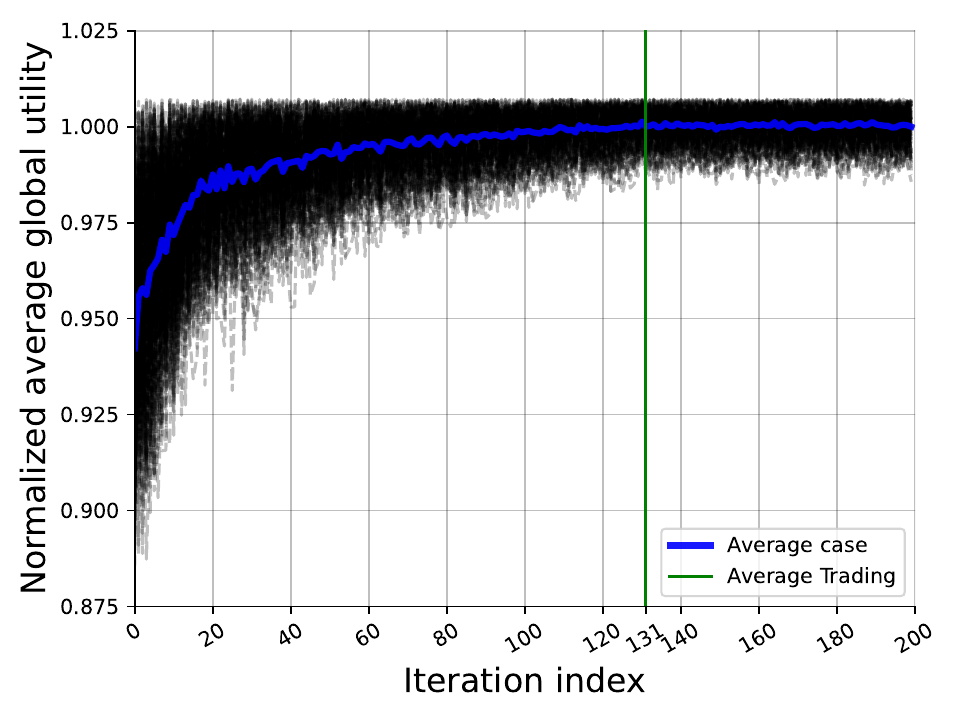}
    \caption{A simulation of UN member countries resource allocation with 10 UN offices for all countries and 10 local offices for each country. The simulation is computed for \(n = 100\) times, each time with slightly different initial conditions. The blue (solid) line presents the average of the simulations. The green vertical line indicates the average duration taken to coverage into the Trading equilibria.}
    \label{fig:converage_rate}
\end{figure}

The proposed model's contribution extends beyond the enhancement of private and global utility, which may not necessarily represent its most significant achievement. Equally important outcomes of the model include the stabilization of the UN budget and the equitable distribution of financial responsibility for UN activities, underscoring its broader impact.

In addition, we computed the number of "free riders" - i.e., agents that did not contribute any resource to the production of the public good\footnote{The free rider phenomenon occurs when individuals or entities benefit from a resource or service without contributing to its cost. For example, in the context of a Nash equilibrium, this phenomenon arises in situations where the equilibrium strategies lead to under-contribution to a public good or resource, due to rational self-interest.}. For the Nash equilibrium of each simulation, we obtain \(2.02 \pm 0.68\) percent of free riders while for the Trading equilibrium is \(1.38 \pm 0.61\). Based on a one-tailed T-test, we show that the number of free riders of the Trading equilibrium is statistically less with \(p < 0.01\).

\section{Discussion and conclusion}
\label{sec:discussion}
The UN's financing mechanisms are crucial for its operational viability and effectiveness in addressing global challenges. The current reliance on earmarked and voluntary contributions poses significant risks to its ability to respond flexibly to crises. Therefore, reforming the funding structure to incorporate diverse and innovative financing strategies is essential for maximizing the UN's capacity to fulfill its mandate in a rapidly changing global landscape.

The present study has sought to extend the analysis of global public goods provision within the framework of the United Nations, applying a novel Trading Equilibrium model to simulate and enhance the financing structure of global public goods. The simulation results reveal significant implications for international cooperation, particularly in providing global public goods through multilateral institutions such as the UN.

The findings underscore that a transition from a Nash equilibrium to a Trading equilibrium yields considerable improvements in global utility. In the Nash equilibrium, each member state acts in its self-interest, contributing suboptimally to the financing of global public goods, resulting in uneven benefits across nations. The Trading equilibrium, by contrast, fosters a cooperative framework where member states’ contributions are aligned with their utility, leading to more efficient allocation of resources. This supports the theoretical assertion that cooperative solutions, modeled through mechanisms such as personalized Lindahl prices, enhance the overall welfare of the international system.

In practical terms, this model can potentially address the inherent inefficiencies associated with free-riding behavior and the uneven financial burden observed in the current UN funding structure. By tailoring contributions to each state’s economic capacity and the benefits they derive, the model mitigates the risk of countries underfunding critical UN programs, particularly in sectors such as peacekeeping and humanitarian aid, which are often underfunded due to voluntary, discretionary contributions.

The implications of this study are profound for future policy directions aimed at reforming the UN’s financing mechanisms. The current system, which heavily relies on voluntary contributions, faces challenges in terms of predictability and equity. Earmarked contributions, as highlighted by other studies, skew the allocation of resources toward the priorities of wealthier donor states, potentially undermining multilateralism. The Trading equilibrium model presents an alternative that aligns more closely with the principle of burden-sharing, while maintaining the autonomy of member states in deciding their contribution levels based on their national interests.

The simulation results suggest that implementing a cooperative framework for UN financing would not only enhance the predictability and sustainability of funding but also ensure that the global public goods provided by the UN benefit all member states more equitably. Such a framework would also support the UN in fulfilling its mandate more effectively, particularly in areas that are traditionally underfunded or politically sensitive.

While the Trading Equilibrium model presents promising outcomes, it is important to note certain limitations of this study. Firstly, the model assumes that all member states are rational actors, willing to cooperate under the right conditions. However, in practice, geopolitical considerations, domestic politics, and differing priorities among member states can hinder cooperation. Future research could explore how such political dynamics might affect the implementation of this model in real-world settings.  In reality, individuals and countries operate with bounded rationality, meaning their decisions are often based on incomplete information and limited processing power. This assumption, while necessary for tractability in the current model, diverges from real-world behavior, where strategic decisions regarding public goods provision are made under uncertainty and with constraints on data availability and computational resources. Future research should aim to relax these assumptions to better reflect the decision-making processes of real-world agents. One possible direction is to incorporate elements of bounded rationality or imperfect information into the model, allowing for the exploration of how these factors impact the convergence to Nash or Trading equilibria. In particular, agent-based simulations that account for different information structures and learning mechanisms could provide valuable insights into how equilibria emerge in more realistic public goods economies.

Moreover, the model's focus on financial contributions does not fully capture other forms of participation and cooperation that are critical to the UN's effectiveness, such as political support, operational capacity, and diplomatic engagement. Expanding the model to account for these non-monetary contributions could provide a more comprehensive understanding of how to optimize global public goods provision in multilateral contexts. To this end, in the countrys' weight fitting procedure (see Section 4), the data distinguish between assessed,
voluntary core (un-earmarked) and voluntary non-core (earmarked) contributions is not made lacking such resolution in the data. Thus, if a country only has assessed contributions to a UN program, this is probably forcing them to spend more than they would have wanted to, which not nessasarly indicates the contry's utility from such spend. 

Taken jointly, the study highlights the potential for significant improvements in the efficiency and equity of UN financing through the implementation of a Trading Equilibrium model. By aligning contributions with both economic capacity and national benefits, the model offers a pathway to more sustainable and cooperative international financing mechanisms. However, further research is needed to account for the complexities of political and non-monetary factors in the international system. The results of this study provide a solid foundation for policymakers and scholars alike to rethink and innovate the ways in which global public goods are financed in the 21st century.

The reliance on voluntary contributions often results in unpredictable and politically influenced funding, limiting the UN’s ability to plan effectively. A shift toward obligatory contributions, structured around personalized pricing models like Lindahl Equilibrium, would stabilize funding and enhance long-term resource allocation for critical global programs.

\section*{Statements and Declarations}
\subsection*{Competing Interests}
The authors have no relevant financial or non-financial interests to disclose.

\subsection*{Funding}
No funds, grants, or other support was received.

\subsection*{Data availability}
The data used in this study is presented in the paper with the relevant references.


\subsection*{Author Contributions} The authors contributed equally.\\ Conceptualization, formal analysis and investigation, validation, original draft, and manuscript editing were performed by Labib Shami. \\ Formal analysis and investigation, software, visualization, original draft, and manuscript editing were performed by Teddy Lazebnik.

\bibliographystyle{chicago}
\bibliography{Trading}

@article{intro_1,
    author = {Kindleberger, C. P.},
    title = "{Dominance and Leadership in the International Economy: Exploitation, Public Goods, and Free Rides}",
    journal = {International Studies Quarterly},
    volume = {25},
    number = {2},
    pages = {242-254},
    year = {1981},
}

@article{intro_2,
title = {Privately provided public goods in a large economy: The limits of altruism},
journal = {Journal of Public Economics},
volume = {35},
number = {1},
pages = {57-73},
year = {1988},
author = {Andreoni, J.},
}

@article{intro_3,
title = {Economic description of tolerance in a society with asymmetric social cost functions},
journal = {Economic research - Ekonomska istraživanja},
volume = {31},
number = {1},
pages = {2548-2593},
year = {2019},
author = {Shi, Y.},
}

@article{intro_4,
    author = {Mailath, G. J. and Postlewaite, A.},
    title = "{Asymmetric Information Bargaining Problems with Many Agents}",
    journal = {The Review of Economic Studies},
    volume = {57},
    number = {3},
    pages = {351-367},
    year = {1990},
}

@article{intro_5,
Author = {Fischbacher, U. and Gachter, S.},
Title = {Social Preferences, Beliefs, and the Dynamics of Free Riding in Public Goods Experiments},
Journal = {American Economic Review},
Volume = {100},
Number = {1},
Year = {2010},
Month = {March},
Pages = {541-556},
}

@ARTICLE{intro_6,
  author={Tesfatsion, L.},
  journal={Artificial Life}, 
  title={Agent-Based Computational Economics: Growing Economies From the Bottom Up}, 
  year={2002},
  volume={8},
  number={1},
  pages={55-82},
}

@article{intro_7,
title = {Varieties of agents in agent-based computational economics: A historical and an interdisciplinary perspective},
journal = {Journal of Economic Dynamics and Control},
volume = {36},
number = {1},
pages = {1-25},
year = {2012},
author = {Chen, S-H.},
}

@article{intro_8,
title={Microfounded tax revenue forecast model with heterogeneous population and genetic algorithm approach},
  author={Alexi, Ariel and Lazebnik, Teddy and Shami, Labib},
  journal={Computational Economics},
  pages={1--30},
  year={2023},
  publisher={Springer}
}

@article{perets2012trading,
  title={Trading equilibrium in a public good economy with smooth preferences and a mixed measure space of consumers},
  author={Perets, Hovav and Shitovitz, Benyamin and Spiegel, Menahem},
  journal={Journal of Mathematical Economics},
  volume={48},
  number={3},
  pages={163--169},
  year={2012},
  publisher={Elsevier}
}

@article{foley1970lindahl,
  title={Lindahl's Solution and the Core of an Economy with Public Goods},
  author={Foley, Duncan K},
  journal={Econometrica: Journal of the Econometric Society},
  pages={66--72},
  year={1970},
  publisher={JSTOR}
}

@article{shitovitz1998cournot,
  title={Cournot--{N}ash and {L}indahl equilibria in pure public good economies},
  author={Shitovitz, Benyamin and Spiegel, Menahem},
  journal={Journal of Economic Theory},
  volume={83},
  number={1},
  pages={1--18},
  year={1998},
  publisher={Elsevier}
}

@article{coase1960problem,
  title={The Problem of Social Cost},
  author={Coase, Ronald H},
  journal={Journal of Law and Economics},
  volume={3},
  pages={1--44},
  year={1960}
}

@article{champsaur1975share,
  title={How to share the cost of a public good?},
  author={Champsaur, Paul},
  journal={International Journal of Game Theory},
  volume={4},
  number={3},
  pages={113--129},
  year={1975},
  publisher={Springer}
}

@article{peleg1986proof,
  title={A proof that the core of an ordinal convex game is a von {N}eumann-{M}orgenstern solution},
  author={Peleg, Bezalel},
  journal={Mathematical Social Sciences},
  volume={11},
  number={1},
  pages={83--87},
  year={1986},
  publisher={Elsevier}
}

@article{shitovitz2001stable,
  title={Stable provision vs. {C}ournot--{N}ash equilibria in pure public good economies},
  author={Shitovitz, Benyamin and Spiegel, Menahem},
  journal={Journal of Public Economic Theory},
  volume={3},
  number={2},
  pages={219--224},
  year={2001},
  publisher={Wiley Online Library}
}

@article{samuelson1954pure,
  title={The pure theory of public expenditure},
  author={Samuelson, Paul A},
  journal={The review of economics and statistics},
  pages={387--389},
  year={1954},
  publisher={JSTOR}
}

@book{galbraith1998affluent,
  title={The affluent society},
  author={Galbraith, John Kenneth},
  year={1998},
  publisher={Houghton Mifflin Harcourt}
}

@book{agents_1,
    author = {Deguchi, H.},
    year = {2004},
    title = {Gaming Simulation and the Dynamics of a Virtual Economy},
    booktitle={Economics as an Agent-Based Complex System},
    publisher   = {Springer}
}

@article{agents_2,
    author = {Heckbert, S. and Baynes, T. and Reeson, A.},
    year = {2010},
    title = {Agent-based modeling in ecological economic},
    journal={Annals of the New York Academy of Sciences},
    volume = {1185},
    pages = {39-63}
}

@inproceedings{agents_3,
  title={Agent-based explanations in {AI}: Towards an abstract framework},
  author={Ciatto, G. and Schumacher, M. I. and Omicini, A. and Calvaresi, D.},
  booktitle={International Workshop on Explainable, Transparent Autonomous Agents and Multi-Agent Systems},
  pages={3-20},
  year={2020},
  publisher={Springer}
}

@INPROCEEDINGS{agent_finite_state_machine,
    author = {Sakellariou, I.},
    year = {2002},
    title = {Agent based Modelling and Simulation using State Machines},
    booktitle={Proceedings of the 2nd International Conference on Simulation and Modeling Methodologies, Technologies and Applications},
    pages = {270-279}
}

@article{gradient_descent_method,
  title={The method of steepest descent for non-linear minimization problems},
  author={Curry, Haskell B},
  journal={Quarterly of Applied Mathematics},
  volume={2},
  number={3},
  pages={258--261},
  year={1944}
}

@article{simplex,
title = {An overview on the simplex algorithm},
journal = {Applied Mathematics and Computation},
volume = {210},
number = {2},
pages = {479-489},
year = {2009},
author = {Nabli, H.},
}

@article{dijkstra2020pareto,
  title={Pareto-efficient solutions for shared public good provision: Nash bargaining versus exchange-matching-lindahl},
  author={Dijkstra, Bouwe R and Nentjes, Andries},
  journal={Resource and Energy Economics},
  volume={61},
  pages={101179},
  year={2020},
  publisher={Elsevier}
}

@article{chen2018collective,
  title={Collective action in an asymmetric world},
  author={Chen, Cuicui and Zeckhauser, Richard},
  journal={Journal of Public Economics},
  volume={158},
  pages={103--112},
  year={2018},
  publisher={Elsevier}
}

@article{danziger1976graphic,
  title={A graphic representation of the {N}ash and {L}indahl equilibria in an economy with a public good},
  author={Danziger, Leif},
  journal={Journal of Public Economics},
  volume={6},
  number={3},
  pages={295--307},
  year={1976},
  publisher={Elsevier}
}

@book{roberts1989lindahl,
  title={Lindahl equilibrium},
  author={Roberts, John},
  year={1989},
  publisher={Springer}
}

@article{kim1984free,
  title={The free rider problem: Experimental evidence},
  author={Kim, Oliver and Walker, Mark},
  journal={Public choice},
  volume={43},
  number={1},
  pages={3--24},
  year={1984},
  publisher={Springer}
}

@article{hampton1987free,
  title={Free-rider problems in the production of collective goods},
  author={Hampton, Jean},
  journal={Economics \& Philosophy},
  volume={3},
  number={2},
  pages={245--273},
  year={1987},
  publisher={Cambridge University Press}
}

@article{hardin1968the,
  title={The Tragedy of the Commons},
  author={Hardin, Garrett},
  journal={Science},
  volume={162},
  number={3859},
  pages={1243--1248},
  year={1968}
}

@article{nash1950equilibrium,
  title={Equilibrium points in n-person games},
  author={Nash Jr, John F},
  journal={Proceedings of the national academy of sciences},
  volume={36},
  number={1},
  pages={48--49},
  year={1950},
  publisher={National Acad Sciences}
}

@article{yu2016equivalence,
  title={Equivalence results between {N}ash equilibrium theorem and some fixed point theorems},
  author={Yu, Jian and Wang, Neng-Fa and Yang, Zhe},
  journal={Fixed Point Theory and Applications},
  volume={2016},
  pages={1--10},
  year={2016},
  publisher={Springer}
}

@inproceedings{arrow1951extension,
  title={An extension of the basic theorems of classical welfare economics},
  author={Arrow, Kenneth J},
  booktitle={Proceedings of the second Berkeley symposium on mathematical statistics and probability},
  volume={2},
  pages={507--533},
  year={1951},
  organization={University of California Press}
}

@article{mcginty2013public,
  title={Public goods provision by asymmetric agents: experimental evidence},
  author={McGinty, Matthew and Milam, Garrett},
  journal={Social Choice and Welfare},
  volume={40},
  pages={1159--1177},
  year={2013},
  publisher={Springer}
}

@article{deneulin2007public,
  title={Public goods, global public goods and the common good},
  author={Deneulin, S{\'e}verine and Townsend, Nicholas},
  journal={International journal of social economics},
  volume={34},
  number={1/2},
  pages={19--36},
  year={2007},
  publisher={Emerald Group Publishing Limited}
}

@article{wyrobek2022public,
  title={Public goods in economics and other social sciences},
  author={Wyrobek, Joanna and Rosiek, Ksymena and Ni{\.z}nik, Joanna},
  journal={Public Goods and the Fourth Industrial Revolution},
  pages={35--69},
  year={2022}
}

@article{labib_ask,
  title={On Weierstrass extreme value theorem},
  author={Martínez-Legaz, J. E.},
  journal={Optimization Letters},
  pages={391-393},
  volume = {8},
  year={2014}
}

@article{cornes2007aggregative,
  title={Aggregative public good games},
  author={Cornes, Richard and Hartley, Roger},
  journal={Journal of Public Economic Theory},
  volume={9},
  number={2},
  pages={201--219},
  year={2007},
  publisher={Wiley Online Library}
}

@article{pandemic_important,
    author = {Conti, A. A.},
    title = "{Historical and methodological highlights of quarantine measures: from ancient plague epidemics to current coronavirus disease (COVID-19) pandemic}",
    journal = {Acta bio-medica : Atenei Parmensis},
    volume = {91},
    number = {2},
    pages = {226-229},
    year = {2020}
}

@online{who_data,
    author = "WHO",
    title = "WHO Coronavirus Disease (COVID-19) Dashboard",
    year = {2022},
    url  = "https://covid19.who.int/",
    addendum = "(accessed: 27.03.2022)"
}

@article{close_1,
title = {Evaluation of the effect of border closure on COVID-19 incidence rates across nine African countries: an interrupted time series study},
journal = {Trans R Soc Trop Med Hyg},
year = {2021},
author = {Emeto, T. I. and Alele, F. O. and Ilesanmi, O.S.},
}

@article{close_2,
title = {Analysing governmental response to the COVID-19 pandemic},
journal = {Journal of Oral Biology and Craniofacial Research},
year = {2021},
author = {Imtyaz, A. and Haleem, A. and Javaid, M.},
}

@ARTICLE{close_bad,
AUTHOR={Zhang, H.},
TITLE={Challenges and Approaches of the Global Governance of Public Health Under COVID-19},
JOURNAL={Frontiers in Public Health},
VOLUME={9},
YEAR={2021},
}

@article{better_together,
  title={Health security capacities in the context of COVID-19 outbreak: an analysis of International Health Regulations annual report data from 182 countries},
  author={Kandel, N. and Chungong, S. and Omaar, A. and Xing, J.},
  journal={The Lancet},
  volume={395},
  number={10229},
  pages={1047-1053},
  year={2020}
}

@article{pandemic_cool,
  title={The Corona-Pandemic: A Game-Theoretic Perspective on Regional and Global Governance},
  author={Caparros, A. and Finus, M.},
  journal={Environmental and Resource Economics},
  volume={76},
  number={4},
  pages={913-927},
  year={2020}
}

@article{coleman2017extending,
  title={Extending UN peacekeeping financing beyond UN peacekeeping operations? The prospects and challenges of reform},
  author={Coleman, Katharina P},
  journal={Global Governance},
  volume={23},
  pages={101},
  year={2017},
  publisher={HeinOnline}
}

@article{jenks2013united,
  title={United Nations Development at a crossroads},
  author={Jenks, Bruce and Jones, Bruce},
  journal={New York: New York University, Center on International Cooperation},
  year={2013}
}

@misc{baumann2024financing,
  title={Financing the United Nations: Status quo, challenges, and reform options},
  author={Baumann, MO and Haug, S},
  year={2024},
  publisher={Friedrich Ebert Stiftung (FES) and German Institute of Development and~…}
}

@article{haug2022international,
  title={International organizations and differentiated universality: reinvigorating assessed contributions in United Nations funding},
  author={Haug, Sebastian and Gulrajani, Nilima and Weinlich, Silke},
  journal={Global Perspectives},
  volume={3},
  number={1},
  pages={39780},
  year={2022},
  publisher={University of California Press}
}

@incollection{10.1093/oso/9780198877936.003.0004,
    author = {Graham, Erin R.},
    isbn = {9780198877936},
    title = "{85Voluntary Funding and Financial Crisis}",
    booktitle = "{Transforming International Institutions: How Money Quietly Sidelined Multilateralism at The United Nations}",
    publisher = {Oxford University Press},
    year = {2023},
    month = {07},
    doi = {10.1093/oso/9780198877936.003.0004},
    url = {https://doi.org/10.1093/oso/9780198877936.003.0004},
    eprint = {https://academic.oup.com/book/0/chapter/413255641/chapter-pdf/51059246/oso-9780198877936-chapter-4.pdf},
}

@article{graham2015money,
  title={Money and multilateralism: how funding rules constitute IO governance},
  author={Graham, Erin R},
  journal={International Theory},
  volume={7},
  number={1},
  pages={162--194},
  year={2015},
  publisher={Cambridge University Press}
}

@article{campos2018political,
  title={The political economy of UN System operational activities: overcoming the bilateralization of multilateralism through pooled funds?},
  author={Campos, Luciana R},
  journal={Mon{\c{c}}{\~o}es: Revista de Rela{\c{c}}{\~o}es Internacionais da UFGD},
  volume={7},
  number={13},
  pages={83--115},
  year={2018}
}

@article{browne2017vertical,
  title={Vertical funds: new forms of multilateralism},
  author={Browne, Stephen},
  journal={Global Policy},
  volume={8},
  pages={36--45},
  year={2017},
  publisher={Wiley Online Library}
}

@article{abdenur2019peacekeeping,
  title={UN peacekeeping in a multipolar world order: Norms, role expectations, and leadership},
  author={Abdenur, Adriana Erthal},
  journal={United nations peace operations in a changing global order},
  pages={45--65},
  year={2019},
  publisher={Springer International Publishing}
}

@incollection{reinsberg2015rise,
  title={The rise of multi-bi aid and the proliferation of trust funds},
  author={Reinsberg, Bernhard and Michaelowa, Katharina and Eichenauer, Vera Z},
  booktitle={Handbook on the economics of foreign aid},
  pages={527--554},
  year={2015},
  publisher={Edward Elgar Publishing}
}

@incollection{shephard_fare_1974,
  author    = {Shephard, R. W. and Färe, R.},
  title     = {The Law of Diminishing Returns},
  booktitle = {Production Theory: Proceedings of an International Seminar Held at the University at Karlsruhe May–July 1973},
  pages     = {287--318},
  publisher = {Springer Berlin Heidelberg},
  address   = {Berlin, Heidelberg},
  year      = {1974},
}

\section*{Appendix}

\subsection*{The Nash provision in our economy}
\label{sec:2}
Below we present a detailed proof of the uniqueness of the Nash equilibrium allocation.

\begin{theorem}
In our economy, for each \(i\in \mathscr{N}\) and \(k \in \mathscr{K}\) we have \(x_i^{k^*}>0\), where \(x_i^{k^*}>0\) is agent \(i\)'s consumption of the private good \(k\) at the unique Nash equilibrium allocation.
\end{theorem}
\begin{proof}
Since each \(i\in \mathscr{N}\) is endowed with a strictly positive amount \(\omega_i=(\omega_i^1,...,\omega_i^K)\in R_{++}^K\) of the private goods and for each agent, \(i\in \mathscr{N}\), the preference relation is defined over his consumption bundles in \(R_+^{K+L}\), where all commodities are strictly desired over \(R_{++}^{K+L}\) and agents are indifferent between all bundles on the boundary of \(R_+^{K+L}\), then by the individual rationality condition we get that each \(i\in \mathscr{N}\) will prefer the allocation in which he contributes half of any quantity of private good in his possession (and distributes the total contribution equally among the public goods) and all the other agents make a non-negative contribution from each private good, rater than an allocation in which there exists \(k\in \mathscr{K}\) for whom \(x_i^k=0\). Since for each \(i\in \mathscr{N}\) we have \(U_i(s^*) \ge U_i(s_i,s_{-i}^*)\), especially when \(s_i=\Big(\frac{\omega_i^1}{2},...,\frac{\omega_i^K}{2}; \frac{\sum_{k\in \mathscr{K}}(\frac{\omega_i^k}{2})}{L},...,\frac{\sum_{k\in \mathscr{K}}(\frac{\omega_i^k}{2})}{L}\Big)\in R_{++}^{K+L}\), it follows that for each \(i\in \mathscr{N}\) and \(k \in \mathscr{K}\) we have \(x_i^{k^*}>0\).
\end{proof}

\begin{theorem}
Our finite economy admits a unique Nash allocation. We denote it by \(((x_i^*)_{i\in \mathscr{N}},Y^*)\). 
\end{theorem}

\begin{proof}
Assume by negation that there are two Nash equilibria \(((x_i^*)_{i\in \mathscr{N}},Y^*)\) and \(((x_i^{**})_{i\in \mathscr{N}},Y^{**})\). Thus, there may be three cases:
\begin{enumerate}
\item
For each \(i\in \mathscr{N}\) and \(k \in \mathscr{K}\) we have \(x_i^{k^*}=x_i^{k^{**}}\) and there exists \(l\in \mathscr{L}\) for whom \(Y_l^* \neq Y_l^{**}\) is met. \\
Assume, w.l.o.g., that for this specific \(l\in \mathscr{L}\) we have \(Y_l^* >Y_l^{**}\). Thus, there exists \(i\in \mathscr{N}\) for whom \(y_i^{l^*}>0\) and for some \(\grave{k} \in \mathscr{K}\) we have \(0<x_i^{\grave{k}^*}<\omega_i^{\grave{k}}\). Hence, by the interior solution condition for Nash equilibrium, we get \(MRS_{Y_l - x_i^{\grave{k}}}^i(x_i^*,Y^*)=1\). Moreover, since for each \(i\in \mathscr{N}\) and \(k \in \mathscr{K}\) we have \(x_i^{k^*}=x_i^{k^{**}}\), we also get \(0<x_i^{\grave{k}^{**}}<\omega_i^{\grave{k}}\) and thus \(MRS_{Y_l - x_i^{\grave{k}}}^i(x_i^{**},Y^{**})=1\). However, since \(Y_l^* >Y_l^{**}\) and \(x_i^{k^*}=x_i^{k^{**}}\), it follows that, by the strict ordinal normality assumption, we must have \(MRS_{Y_l - x_i^{\grave{k}}}^i(x_i^*,Y^*)<MRS_{Y_l - x_i^{\grave{k}}}^i(x_i^{**},Y^{**})=1\). A contradiction.

\item
For each \(l\in \mathscr{L}\) we have \(Y_l^* =Y_l^{**}>0\) and there exists \(i\in \mathscr{N}\) for whom \(x_i^{k^*}\neq x_i^{k^{**}}\) is met. \\
By Condition \ref{Nash2} of Definition \ref{Nashdef} we get \(\sum_{l\in \mathscr{L}}Y_l^*=\sum_{i\in \mathscr{N}}\Big(\sum_{k\in \mathscr{K}}(\omega_i^k-x_i^{k^*})\Big)\). Thus, there exists \(k\in \mathscr{K}\) and \(i\in \mathscr{N}\) for whom we have \(x_i^{k^*}<x_i^{k^{**}}\le \omega_i^k \) so that for some \(l\in \mathscr{L}\) we get \(MRS_{Y_l - x_i^k}^i(x_i^*,Y^*)=1\). However, since for this specific \(l\in \mathscr{L}\) we have \(Y_l^* =Y_l^{**}\) and \(x_i^{k^*}<x_i^{k^{**}}\) it follows that, by the strict ordinal normality assumption, we must have \(MRS_{Y_l - x_i^{k}}^i(x_i^*,Y^*)<MRS_{Y_l - x_i^{k}}^i(x_i^{**},Y^{**})\le1\). A contradiction.

\item
There exists \(i\in \mathscr{N}\) and \(l\in \mathscr{L}\) for whom \(x_i^*\neq x_i^{**}\) and \(Y_l^* \neq Y_l^{**}\), respectively. \\
Assume, w.l.o.g., that for this specific \(l\in \mathscr{L}\) we have \(Y_l^* >Y_l^{**}\). Thus, there exists \(i\in \mathscr{N}\) for whom \(y_i^{l^*}>0\) and for some \(\grave{k} \in \mathscr{K}\) we have \(0<x_i^{\grave{k}^*}<\omega_i^{\grave{k}}\). Hence, by the interior solution condition for Nash equilibrium, we get \(MRS_{Y_l - x_i^{\grave{k}}}^i(x_i^*,Y^*)=1\). Now, if for each \(i\in \mathscr{N}\) and \(k \in \mathscr{K}\) we have \(x_i^{k^*}>x_i^{k^{**}}\), then by \(\sum_{l\in \mathscr{L}}Y_l^*=\sum_{i\in \mathscr{N}}\Big(\sum_{k\in \mathscr{K}}(\omega_i^k-x_i^{k^*})\Big)\) there must be \(\grave{l}\in \mathscr{L}\) for whom \(Y_{\grave{l}}^* <Y_{\grave{l}}^{**}\) where, by the strict ordinal normality assumption, for each \(i\in \mathscr{N}\) and \(k \in \mathscr{K}\) we get \(MRS_{Y_{\grave{l}} - x_i^{k}}^i(x_i^{**},Y^{**})<MRS_{Y_{\grave{l}} - x_i^{k}}^i(x_i^{*},Y^{*})\le1\). However, since \(Y_{\grave{l}}^{**}>0\) it follows that for some \(i\in \mathscr{N}\) and \(k \in \mathscr{K}\) we must have \(MRS_{Y_l - x_i^k}^i(x_i^*,Y^*)=1\). A contradiction. Thus, if we assume, w.l.o.g, that for this specific \(l\in \mathscr{L}\) we have \(Y_l^* >Y_l^{**}\) then there exist \(i\in \mathscr{N}\) and \(k \in \mathscr{K}\)  for whom \(x_i^{k^*}<x_i^{k^{**}}\le \omega_i^k \) where \(MRS_{Y_l - x_i^k}^i(x_i^*,Y^*)=1\). However, since we have \(Y_l^* >Y_l^{**}\) and \(x_i^{k^*}<x_i^{k^{**}}\le \omega_i^k \), it follows, by the strict ordinal normality assumption, that \(MRS_{Y_l - x_i^{k}}^i(x_i^*,Y^*)<MRS_{Y_l - x_i^{k}}^i(x_i^{**},Y^{**})\le1\). A contradiction.
\end{enumerate}
\end{proof}

\subsection*{Model computation example}
For a numerical illustration of this paper, consider the following example. Assume all agents share the utility function:
 \[u^i(x_i^1,x_i^2,Y_1, Y_2)=\sqrt{x_i^1}\sqrt{x_i^2}Y_1 Y_2\]  \[i=1,2\]
The initial endowment of each agent is:
\[\omega_1=(3,2) \, \, , \, \,  \omega_2=(2,3)\]
By the interior solution condition for Nash equilibrium, \(MRS_{Y_l - x_i^k}^i(x_i^*,Y^*)=1\), we get \(\forall i \in \mathscr{N}\):

\[x_i^1=x_i^2=0.5Y_1\]
\[x_i^1=x_i^2=0.5Y_2\]
\[Y_1=Y_2=y_1^1+y_2^1=y_1^2+y_2^2\]
Substituting these results in the budget constraint \(\sum_{k\in \mathscr{K}} x_i^k+\sum_{l\in \mathscr{L}}y_i^l \le \sum_{k\in \mathscr{K}} \omega_i^k\) will yield:

\[0.5Y_1+0.5Y_1+y_1^1+y_1^2=5\]
\[0.5Y_1+0.5Y_1+y_2^1+y_2^2=5\]
Vertical summation will yield \(4Y_1=10\), thus \(Y_1^*=2.5\) and the unique Nash equilibrium allocation is:
\[(x^*_1,x^*_2;Y^*_1,Y^*_2)=(1.25,1.25,1.25,1.25;2.5,2.5)\]
And the utilities are:
\[u^*_1=u^*_2=7.8125\]

Note that this unique allocation results from infinite possible combinations of Nash strategies. That is, the contribution of the first agent to the production of the first public good depends on and complements the contribution of the second agent to the production of the same public good. This is the case for the production of the second public good. Hence, there may be infinite combinations of the agents' contributions to the production of the public goods, but the result of all those combinations leads to the same amount of each of the public goods (symmetry of strategies). In the case of our example, the final result is \(Y_1^*=Y_2^*=2.5\).

To calculate the Lindahl equilibrium allocation we will use the general rule of \cite{samuelson1954pure}. That is, for each public good one has \( \sum_{i \in \mathscr{N}} MRS_{Y_l-x_i^k}^{i}=RPT_{Y_l-x_i^k}\), where \(RPT_{Y_l-x_i^k}=1\) is the rate of technological substitution in production between public and private goods. Thus:
\[\frac{x_1^1}{0.5Y_1}+\frac{x_2^1}{0.5Y_1}=1 \Rightarrow x_1^1+x_2^1=0.5Y_1\]

\[\frac{x_1^2}{0.5Y_1}+\frac{x_2^2}{0.5Y_1}=1 \Rightarrow x_1^2+x_2^2=0.5Y_1\]

\[\frac{x_1^1}{0.5Y_2}+\frac{x_2^1}{0.5Y_2}=1 \Rightarrow x_1^1+x_2^1=0.5Y_2\]

\[\frac{x_1^2}{0.5Y_2}+\frac{x_2^2}{0.5Y_2}=1 \Rightarrow x_1^2+x_2^2=0.5Y_2\]
Which yields \(Y_1=Y_2\). Placing these results in condition 1 of definition 2 ( i.e., \(\sum_{k\in \mathscr{K}} \sum_{i\in \mathscr{N}} \bar{\bar{x}}_i^k+ \sum_{l\in \mathscr{L}}\bar{\bar{Y}}_l=\sum_{k\in \mathscr{K}} \sum_{i\in \mathscr{N}} x_i^{k*}+\sum_{l\in \mathscr{L}} Y_l^* =\sum_{k\in \mathscr{K}} \sum_{i\in \mathscr{N}} \omega_i^k \) ) will yield \(Y_1 = Y_2 =1\frac{1}{3}\), and the Trading equilibrium allocation will be:
 \[(\bar{\bar{x}}_1,\bar{\bar{x}}_2;\bar{\bar{Y_1}},\bar{\bar{Y_2}};\bar{\bar{p}}_1,\bar{\bar{p}}_2)=(\frac{5}{6},\frac{5}{6},\frac{5}{6},\frac{5}{6};3\frac{1}{3},3\frac{1}{3};\frac{1}{2},\frac{1}{2},\frac{1}{2},\frac{1}{2})\]
And the utilities are:
\[\bar{\bar{u}}_1=\bar{\bar{u}}_2=9.26\]
Note that for all \( i \in \mathscr{N}\) we get:
 \[\bar{\bar{u}}_i>u^*_i \]
Thus, the Trading equilibrium allocation strictly
Pareto dominates (utility-wise) the Nash allocation. That is, all
agents would prefer to replace their Nash bundle with their
Trading equilibrium bundle.

\end{document}